\documentclass[12pt,a4paper]{article}

\usepackage[british]{babel}

\usepackage[a4paper,top=2cm,bottom=2cm,left=2.5cm,right=2.5cm,marginparwidth=1.75cm]{geometry}

\usepackage[sorting=nyt]{biblatex} 
\usepackage{graphicx} 
\usepackage[colorlinks=true, allcolors=blue]{hyperref} 
\usepackage{orcidlink,algorithm,algorithmicx,algpseudocode,listings,chngcntr,booktabs,lipsum,subcaption,authblk,diagbox}
\usepackage[title]{appendix}
\usepackage{booktabs} 
\usepackage{caption}  
\usepackage{threeparttable} 
\usepackage[T1]{fontenc}    
\usepackage{csquotes}       
\usepackage[utf8]{inputenc}

\usepackage{amsmath}
\usepackage{amstext}
\usepackage{amsbsy}
\usepackage{amsthm}
\usepackage{amsfonts}
\usepackage{amssymb}
\usepackage{amscd}
\usepackage{eucal}
\usepackage{enumerate}
\usepackage{units}
\usepackage{centernot}
\usepackage{verbatim}
\usepackage{xcolor}
\usepackage{mathtools}
\usepackage{bbm}
\usepackage{mathrsfs}
\usepackage{mdwlist}
\usepackage{geometry}

\usepackage{bm}

\theoremstyle{plain}
\newtheorem{thm}{Theorem}[section]
\newtheorem*{thm*}{Theorem}
\newtheorem{lem}[thm]{Lemma}
\newtheorem*{lem*}{Lemma}
\newtheorem{cor}[thm]{Corollary}
\newtheorem{prop}[thm]{Proposition}
\newtheorem*{prop*}{Proposition}
\theoremstyle{definition}
\newtheorem{defn}[thm]{Definition}
\newtheorem*{defn*}{Definition}

\def\R{{\mathbb R}}
\def\footnoterule{\hrule \kern2.6pt}

\def\0T{[\,0,T]}

\newcommand{\N}{\mathbb{N}}

\newcommand{\norm}[1] {\left\| #1 \right\| }

\usepackage{setspace}
\usepackage{titlesec}
\titleformat{\section} 
  {\normalfont\Large\bfseries}{\thesection.}{1em}{}
  
\usepackage{float}   
\usepackage{caption} 


\title{Existence and Optimality of Envy-Free random allocations}

\author[1]{Anna Vakarova}

\affil[1]{\small UC Berkeley}

\date{}  

\begin{document}
\maketitle

\begin{abstract}
I provide a unified framework to establish the existence of a weak Pareto efficient, envy-free allocation in general settings: random allocations are probability measures on a compact metric space, and preferences of agents are represented by continuous, concave utility function on the space of probability measures.

The generality of my setting nests the existence results for small spaces with indivisibles - the list of prominent applications includes the school assignment problem and the house allocation problem.

The technique developed to prove the existence also applies to allocation problems with divisibles, like fair cake-cutting or land-division problems. Here I also show that even when agents' preferences are not atomless, the allocation in question can be represented as a probability measure over partitions with finite support.

Last but not least, I apply the existence result to new allocation problems that no existing framework encompasses. These include allocation of indivisible goods or services over time and allocation of differentiated goods.

\end{abstract}



\section{Introduction}

The literature on the notion fair and efficient allocation is vast and dates back to Foley \cite{Foley1967}. The questions of the existence, approximate solutions, mechanisms, and algorithms to construct one have been studied in different environments, both by economists and mathematicians.

In this diverse literature, different model primitives call for different proof strategies. For example, in a standard Arrow-Debreu setting, under mild assumptions about agents' preferences, Walrasian equilibrium with equal incomes is known to be PE and EF. Therefore, the exercise of establishing the existence of a fair allocation is equivalent to establishing the existence of Walrasian equilibrium. However, this approach fails in the setting of Hylland and Zeckhauser, referred to as HZ throughout this paper \cite{HyllandZeckhauser1979}. Due to the unit-demand constraint, Walrasian equilibrium may not exist. To prove existence of a fair allocation, HZ employ a different strategy - proving the existence of a pseudo-market equilibrium that is PE and envy-free.


In some important applications though, the unit demand constraint can be too restrictive. When generalizing unit demand setting to, say, multi-unit demand setting, the literature faces a challenge in applying the concept of the pseudo-market equilirium, namely, the failure of the Birkoff-von Neumann theorem. The theory of existence of a pseudo-market equilibrium relies on particular features of the setting: every agent gets their marginal lottery, and jointly the lotteries agents get must be feasible. There is no guarantee though that there is a way to distribute the goods that respects both the demand constraints and does not distort agents' marginal lotteries. The application frontier of the approach that relies on the BvN theorem is beautifully characterized by Budish et al (2013) \cite{BudishEtAl2013}.


This paper abstracts from the types of constraints that might be imposed on agents' demands, and even more generally, on the set of feasible allocations, and attempts to provide a unified framework for which the existence result is proved. This framework nests settings commonly found in both economics and mathematics literature.

The primitive of the setting of this paper that allows me to nest other settings is general deterministic consumption spaces. I make this definition precise in Section 3. Other important features of the literature, like finitely many agents and the assumption that utilities are non-transferable, are preserved. The random allocations in the paper are probability measures, with non-convexity of the deterministic consumption space as one potential justification.

In Section 3, I also present different settings that can be found in the literature as special cases of a general deterministic consumption space. In particular, I am able to accomodate different multi-unit demand settings, as described in \cite{MirallesPycia2021} or \cite{BudishEtAl2013}, or the setting where there is pre-specified family of partitions that defines the way the goods can be allocated (\cite{ColeTao2019}). 

Another example of a general deterministic consumption space I am able to accomodate comes from the literature on fair division in mathematics, which studies the problem of allocating a "continuous" commodity. Husseinov and Sagara \cite{HusseinovSagara2013} characterize consumption spaces as "intrinsically infinite-dimensional." Berliant \cite{Berliant1985} also makes a strong case for modeling land explicitly as a subset of the plane rather than fitting it into a finite-dimensional space. In these models, the consumption set of agents usually consists of partitions of a measurable space, which by means of simple measurable functions can be embedded in an infinite-dimensional vector space.

Working with general deterministic consumption spaces goes beyond nesting applications for which the existence result is already established. A new application studies allocation problems of goods or services over time, such as airline time-slot scheduling. Here agents must engage in long-horizon planning over sequences of interdependent decisions rather than static bundles. Moreover, an airline can face a number of involved constraints, such as capacity constraints, time and network constraints (since flights must form feasible itineraries that respect aircraft turnaround times) maintenance requirements, and connectivity across routes. To the best of my knowledge, the existing frameworks are not sufficiently rich to capture the structure of this problem.

Another new application is the problem of allocating heterogeneous, indivisible delivery orders to drivers in a stochastic platform environment. In Section 5, I claim that the model of differentiated goods a-la MasCollel \cite{MasColell1975} captures the main features of the application. Agents choose from a continuum of characteristics of commodities, but each commodity must be consumed in integral amounts.

Both applications fall into the realm of fair division. Time slot allocation in airports is a natural candidate for fairness considerations because airport infrastructure operates as a regulated natural monopoly, where capacity is scarce, access is centrally mediated, and distributional concerns are explicitly recognized in institutional design. By contrast, platform-based order allocation to drivers is also shaped by fairness concerns because subjective perceptions of unequal access to lucrative assignments can directly affect participation, effort provision, and the stability of the labor supply.

To summarize, my contribution is threefold. I establish the existence of weak Pareto efficient and envy-free allocations in the environments with general allocation spaces. It allows me to apply the existence result to new applications that were not studied before and are not directly accomodated by the existing models. The existence result also applies to common problems in both econ and math literature, like land division problem. Here I show that fair and efficient allocation can be chosen to have finite support.




Proving the existence of a fair allocation is essentially an application of a fixed point theorem or equivalent, therefore a tight connection to an equilibrium with equal incomes is natural and commonly found in the literature. Differently from that, the meat of the approach of this paper that allows to establish the result at this level of generality is Sperner's lemma (or equivalently, KKM lemma). The proof closely follows that of Echenique, Miralles, and Zhang \cite{EcheniqueEtAl}. Not relying directly on the existence of equilibrium arguments is what allows me to circumvent failures of the Birkoff-von-Neuman theorem. The proof is structured as follows: I start by characterizing all wPE allocations via Pareto weights. To find the correct vector of weights for envy-freeness to hold, I simplicially subdivide the simplex of weights and color the subsimplexes using a version of the Varian lemma that finds an agent in support who does not envy anyone. The application of Sperner's lemma and the continuity of utility functions yield the desired result.

The paper is structured as follows. Section 2 reviews both the literature on the existence of efficient and envy-free allocations in economics and the cake-cutting literature in mathematics. Sections 3 and 4 present the setup and the proof of the main theorem. Section 5 illustrates the usefullness of the existence result through several new applications. Section 6 concludes.

\section{Literature review}

The problem of reconciling fairness with efficiency in economic allocations has deep roots, but it was Hal Varian’s paper \cite{Varian1974}, that first rigorously formulated the tension between envy-freeness and Pareto efficiency in general equilibrium theory. Varian defined an allocation as envy-free if no individual prefers someone else's allocation to their own. He provided positive results for cases with divisible goods and convex preferences, noting that such allocations often exist in standard economies. However, his work also highlighted the fragility of envy-freeness when assumptions like convexity or divisibility are dropped. Since in the settings with indivisible goods, PE and EF allocations fail to exist in the simplest settings, HZ introduced the idea of randomized allocations.

In the literature on allocating indivisible goods, two main branches exist based on how preferences are represented. In ordinal models, agents have preferences over lotteries/bundles, without assuming expected utility representation. An influential example here is Bogomolnaia and Moulin \cite{BogomolnaiaMoulin2001}. They introduced a notion of ordinal efficiency and a mechanism (the Probabilistic Serial mechanism) that ensures ordinal efficiency and envy-freeness under stochastic assignments.

It is well-illustrated in the literature that working with ordinal preferences as partial information on the cardinal preferences may lead to big welfare losses. For example, Daniel Halpern and Nisarg Shah \cite{HalpernShah2021} characterize the size of a distortion with or without fairness requirements of allocation mechanisms under ordinal preferences.

Natually, much of the theoretical depth in this field lies with the cardinal preferences, where agents have vNM utility functions defined over objects, or, in other words, deterministic allocations, allowing the discussion of cardinal efficiency. 

The canonical model in this area was introduced by Hylland and Zeckhauser \cite{HyllandZeckhauser1979}. They studied the problem of allocating indivisible goods (like seats in classes) under unit demand constraints, meaning each agent can have at most one item. Their key idea was to allow randomized assignments (i.e., lotteries over discrete allocations) and treat them as divisible probabilistic shares of the goods. Using this probabilistic representation, they showed that there exists an allocation that is both ex-ante Pareto efficient and ex-ante envy-free. Note that the unit demand constraint in HZ is crucial, as it allows to show that the set of feasible allocations can be represented as a lottery over the set of deterministic/discrete allocations via the Birkhoff-von-Neuman theorem.

“Designing Random Allocation Mechanisms: Theory and Applications” by Budish, Che, Kojima, and Milgrom \cite{BudishEtAl2013} tackles the incompatibility of efficiency, incentive compatibility, and envy-freeness in settings with indivisible goods. They build upon the HZ model but shift the focus to large markets. Their approximate Competitive Equilibrium from Equal Incomes (A-CEEI) mechanism finds allocations that are approximately EF, PE, and IC under quasi-linear, additive utility functions over bundles of indivisible goods. As the market becomes large, the paper proves that A-CEEI allocations converge to the exact competitive equilibrium outcomes of the divisible goods economy.

This paper is inspired by "Fairness and Efficiency for Allocations with Participation Constraints" by Echenique, Miralles, and Zhang \cite{EcheniqueEtAl}. Their paper extends the HZ framework to settings with participation constraints, where agents cannot be worse off by participating than by opting out. The authors show that under general conditions, there exists a random assignment that is Pareto efficient, justified envy-free, and satisfies the participation constraints. Under different assumptions on the utility functions, their proofs are existence of equilibrium proof, with price-dependent income, and search for the right vector of Pareto weights proof, using the KKM lemma.

Cole and Tao \cite{ColeTao2019} also prove the existence of exact envy-free and Pareto efficient allocations. They relax unit demand constraints and work with the set of feasible allocations that is the set of permutations on the set of partitions of finitely many items.

The problem of finding a fair allocation in particular settings has been studied not only by economists but also by mathematicians. The famous literature on cake cutting or land division traces back to Steinhaus \cite{Steinhaus1948}, who formalized the problem known as the cake-cutting problem, starting the literature on fair division. But this particular definition of fairness, that is, envy-freeness and Pareto efficiency, has not found its way into the literature until Weller \cite{Weller1985}. He proved the existence of an EF and PE allocation in the problem of allocating a measurable space $(\Omega, \mathcal{B})$, where agents' preferences are atomless measures on this space. He also showed that such an allocation is supported as a price equilibrium in the exchange economy with equal incomes.

Akin \cite{Akin1995} builds on Weller and extends the result to the settings where measures are not necessarily atomless. His notion of envy-freeness is also more general, allowing for assigning non-equal weights to agents' evaluations.

Weller's allocation is a partition, since he assumes atomless measures. Akin's allocation is a partition of unity, since he does not. Both of them give a way to construct such an allocation, though: it is supported as a price equilibrium for certain prices/budgets. My contribution to the literature is that even when measures are not atomless, I show that the allocation in question is a lottery over partitions.

Berliant, Thomson, and Dunz \cite{BerliantThomsonDunz1991} prove the existence of an allocation satisfying PE and the stronger notion of group envy-freeness, but the comestible in the paper is a compact subset of $\mathbb{R}^k$ and agents' utilities are absolutely continuous with respect to Lebesgue measure. They also prove that under certain assumptions, envy-freeness implies group envy-freeness and efficiency. For example, their sufficient conditions are applicable in the problem of allocating sub-intervals of $[0,1]$ as in Woodall \cite{Woodall1980}, and the existence of an EF allocation has already been proved.

Hüsseinov and Sagara \cite{HusseinovSagara2013} model a commodity as a measurable space. They relax the convexity assumption of the preferences and prove the existence of a PE and EF allocation under certain assumptions on the utility function of each individual: uniform continuity, strict monotonicity with respect to set inclusion order, and closed partition matrix range of the utility functions.

\section{Set Up}\label{sec3}
There are finitely many agents in the allocation problem, indexed by $j \in J$. The consumption space of a single agent is a metrizable space $(Y,d_Y)$. Without loss, I assume it is identical for all agents. 

\subsubsection*{The space of feasible random allocations}
The set of feasible deterministic allocations is denoted by $X$ and it is a subset of $Y^N$. $\mathcal{B}_X$ denotes the sigma-algebra of Borel measurable sets, and Borel space $(X,\mathcal{B}_X)$ is taken as a primitive of the problem.

The generality of the existence result is primarily driven by the generality of the set $X$. Unlike the existing papers in the fair division literature in economics, I do not assume it is finite (a set of extreme points of a convex polytope), or a subset of a finite-dimensional vector space. Also, my existence result does not exploit its geometric structure.

I do impose two assumptions on $X$. To this end, I define a \textbf{permutation operator} $T_{i \leftrightarrow j}$ on $X$ that maps $y = (y_1, \dots, y_i, \dots, y_j, \dots, y_J)$ into $y_{i \leftrightarrow j} = (y_1, \dots, \mathbf{y_j}, \dots, \mathbf{y_i}, \dots, y_J)$. Note that the permutation operator is continuous in the product topology. Moreover, for any $i,j$, $T_{i \leftrightarrow j}$ is an involution, hence a homeomorphism.

The set $X$ is assumed to be:
\begin{itemize}
    \item compact
    \item \textbf{permutation-invariant}: $T_{i \leftrightarrow j}(X) \subset X \ \text{for all} \ i,j$
\end{itemize}

Permutation invariance is also known as invariance under re-labeling of agents. Note that when reasoning about fairness, it ensures "the equality of opportunity" - no agent is exogenously discriminated against. 

As $X$ is permutation-invariant, $\pi_i(X) = \pi_j(X) \ \forall i,j \in J$, where $\pi:X \rightarrow Y$ is a natural projection map. WLOG, assume that $\pi_i(X) = Y$.

Before proceeding, let me discuss the assumptions of the setting. For the existence result, the compactness of the space of deterministic allocations $X$ is necessary, since the proof relies on the argmax characterization of weak Pareto efficient allocations. Note that if the underlying space is a metrizable topological space, the underlying space is compact iff the space of probability measures on it is compact in weak-* topology. Moreover, to support a technical point of the proof, I require the space of probability measures to be metrizable. A metrizable underlying topological space is compact iff the space of probability measures is compact and metrizable.



The permutation-invariance requirement of the set $X$ is conceptual and NOT without loss of generality. It allows me to use envy-freeness as a criterion of fairness and plays a fundamental role in the proof. Note that if one wants to reason about fairness in more complex settings (for example, the settings with endowments), one way to do so is to restrict access of individuals to certain allocations via agent-specific feasibility constraint, which is not allowed here.

The set of random allocations is denoted by $\Delta(X)$ and it is the set of probability measures on $X$ endowed with weak-* topology. \footnote{For an excellent reference on how this topology is defined, see Aliprantis and Border, 1999, Chapter 15.} Denote an element of this set by $p \in \Delta(X)$. 

Every random allocation $p \in \Delta(X)$, generates a family of marginal distributions denoted by $(p^1, \dots, p^J)$. $p^j$ is a probability measure on $Y$, and for every Borel measurable subset $B \subset Y$, $p^j(B) = p(\pi_j^{-1}(B))$.

\subsubsection*{Preferences}

Agents’ preferences exhibit no externalities. Agent $j$ consumes their marginal distribution allocation $p^j$.

The preferences of an agent $j$ over random allocations are assumed to be represented by a continuous, \textbf{concave} real-valued utility function $U_j: \Delta(Y) \rightarrow \mathbb{R}$. The notation $U_j(p)$ for a random allocation $p$ means nothing but $U_j(p^j)$ for a corresponding marginal $p^j$.

In particular, the preferences described above nest the Expected Utility representations. The preferences are then represented by a continuous VNM utility index $u^j: X \rightarrow \R$.
In this case, I also assume that $u^j$ depends only on the $j$th coordinate of $X$.
    \[
        U_j(p) = \int_{Y} u^j(y) dp^j
    \]

Another class of preferences that satisfy the assumptions above is maxmin preferences under risk. Imagine a conservative decision-maker who is unsure about state-dependent payoffs. For a closed, convex set of VNM utility indeces $\{u^j: Y \rightarrow \R\}$ denoted by $\mathcal{U}^j$
    \[
        U_j(p^j) = \min_{u^j \in \mathcal{U}^j} \int_{Y} u^j(y) dp^j
    \]

\subsection{Examples of the underlying space $X$}

To highlight the generality of the underlying space of deterministic allocations $X$, let me discuss examples from the literature that are special cases of the setting of this model.

\textit{Example 1}. In the classical setting of the literature (HZ), where the problem is to allocate $N$ discrete items, the set $X$ is the set of all permutation matrices $N \times N$.

\textit{Example 2}. In the setting with multi-unit demand a-la \cite{MirallesPycia2021}, the space of allocations is defined as follows.
There are indivisible objects $w, v \in \Omega = \{1, \ldots, |\Omega|\}$ to be assigned to finitely many agents. Each object $w$ is represented by a number of identical copies $|w| \in \mathbb{N}$. Let
    \[
    S = (|w|)_{w \in \Omega}
    \]
denote the total supply of object copies in the economy. Each agent demands at most $k \in \{1, 2, \ldots\}$ units of various goods in total. The space of deterministic consumptions of an agent $Y$ is a subset of $\{0,1,\ldots,k\}^{\Omega}$. For example, in a course allocation problem, the set  $Y = \{ b \in \{0,1\}^\Omega : \sum_{w \in \Omega} b^w \leq k\}$ for each agent $i$. The space of feasible allocations $X$ is cut out from $\times_{j \in J} Y$ by the feasibility constraint $\sum_{j \in J} y_j \leq S$.

\textit{Example 3}. In \cite{ColeTao2019} there are $m$ items to be allocated to $n$ agents, and $T$ feasible partitions of the set of items $\{D_t\}_{t \in T}$ to do so. A partition $D_t = (D_{1t},...,D_{tn})$ is a vector of subsets of $\{1, \dots, m\}$ such that $D_{ta} \cap D_{tb} = \emptyset$ for all agents $a \neq b$ and $\cup_a D_{ta} \subset \{1,2,...,m\}$. Each partition can be distributed to the players in any way, so long as each agent gets a distinct single element of the partition.
To map this setting into the primitives of this paper, let $Y = 2^{\{1, \dots, m\}}$. $X$ is equal to the set of all permutations of feasible partitions.

\textit{Example 4}. Consumption of differentiated goods a-la \cite{MasColell1975}.
MasCollel describes an exchange economy in which differentiated commodities are available in integral amounts (that is, not perfectly divisible) and the variety of commodities comes from the richness of characteristics. Let $(K,d)$ be a compact metric space, called the space of commodity characteristics. The space of individual commodity bundles is the space of non-negative, bounded Borel measures that are integer valued. The consumption set $Y$ is the set of all individual commodity bundles such that $a(K) \leq \alpha$, where $\alpha$ is a large positive integer. 


Let $\nu$ be a total endowment measure of the economy such that $supp(\nu) = K$. The space of deterministic allocations is cut out of Cartesian product by a standard feasibility constraint:
\[
    X = \{(a_1, \dots, a_n) \ \Big| \ a_j \in Y, \sum_{j \in J} a_j \leq \nu\}
\]

\textit{(Counter) Example 5}. The permutation-invariance assumption on the set $X$ is NOT without loss of generality. Consider a setting from \cite{PaznerSchmeidler1974}, which features two examples of market economies in which Pareto efficient and envy free allocation fails to exist.\footnote{I am grateful to 
Shiran Rachmilevitch for the reference.} There are two consumers $J = \{1,2\}$ and two goods to consume, leisure $l$ and consumption $z$. The set of deterministic consumptions available to consumer $j$ is $Y = \{(l_j,z_j) \in \R^2_+ \ \Big| \ 0 \leq l_j \leq 1\}$. The set $X$ is cut out from $Y^2$ by a linear technological constraint $z_1 + z_2 - (1 - l_1) - \frac{1}{10} (1 - l_2) \leq 0$. That is,
\[
    X = \{(l_1,z_1),(l_2,z_2) \in \R^4_+ \ \Big| \ 0 \leq l_j \leq 1 \ \text{and} \ z_1 + z_2 - (1 - l_1) - \frac{1}{10} (1 - l_2) \leq 0\}
\]

Note that the set $X$ does not satisfy the anonimity requirement: even though the vector $(1,\frac{5}{10}), (0,\frac{5}{10}) \in X$, the vector $(0,\frac{5}{10}), (1,\frac{5}{10})$ is not in $X$.

\subsection{Definitions}\label{subsec1}

\begin{defn}[weak Pareto efficiency]
    Allocation $p$ is \textbf{weak Pareto efficient}(wPE) if there is no allocation $p^\prime$ such that $U_j(p^\prime) > U_j(p)$ for all $j \in J$ (no strong Pareto improvements).
    \end{defn}

\begin{thm} \label{thm:wpe}
    Allocation $p$ is weak Pareto efficient if and only if it maximizes $\sum_{j = 1}^n \lambda_j U_j(p)$ for a nonnegative vector $(\lambda_j)_{j \in J} \in \R^N_+$.
\end{thm}

\begin{proof}
    The "if" direction is trivial. For the "only if" direction, observe that the set $\{(U_1(p), \dots, U_N(p)): p \in \Delta(X)\} \subset \mathbb{R}^N$ is an image of a compact set under a continuous map, hence closed. If utilities are assumed to be linear, the image of utilities is also covex. If utilities are concave, the "hypograph" of $p \rightarrow U(p)$, that is, the set $\{v \in \mathbb{R}^N \ \Big| \ \text{there is} \ p: v \leq U(p)\}$ is convex and closed. If $p$ is weakly Pareto efficient, the corresponding vector of utilities is on the boundary of the image/hypograph, and the supporting hyperplane theorem establishes the result.
\end{proof}

Take a set $B \in \mathcal{B}$. Define $B_{j \xleftrightarrow[]{} k} = \{x_{j \xleftrightarrow[]{} k} \ \text{such that} \ x \in B \}$ as the image of $B$ under the permutation operator $T_{j \xleftrightarrow[]{} k}: X \rightarrow X$. Observe that for every $j,k$ the permutation operator $T_{j \xleftrightarrow[]{}k} $ is a homeomorphism, hence both the operator and its inverse are measurable.

\begin{defn}  
    An \textbf{allocation swap} is a new probability measure on $X$ defined as a push-forward measure of $p$ under the permutation operator $T_{j \xleftrightarrow[]{} k}$: $p_{j \xleftrightarrow[]{} k} (B) = p(T^{-1}_{j \xleftrightarrow[]{} k} (B)) = p(B_{j \xleftrightarrow[]{} k})$.
\end{defn}

Allocation swap can be defined for any $p \in \Delta(X)$, giving rise to a family of exchange operators $L_{j \xleftrightarrow[]{} k}: \Delta(X) \rightarrow \Delta(X)$ defined via $L_{j \xleftrightarrow[]{} k}(p) = p_{j \xleftrightarrow[]{} k} = p \circ T_{j \xleftrightarrow[]{} k}$. 

The marginals of $p$ and $p_{j \leftrightarrow k}$ are naturally related:
\[
    p^j_{j \leftrightarrow k}(B^j) = p_{j \leftrightarrow k}(\pi_j^{-1}(B^j)) = p(\pi_j^{-1}(B^j)_{j \leftrightarrow k}) = p(\pi_k^{-1}(B^j)) = p^k(B^j)
\]

\begin{defn}[Envy]
    Fix a random allocation $p$. Agent $j$ envies agent $k$ if $U_j(p_{j \xleftrightarrow[]{} k}) > U_j(p)$, or in other words, $U_j(p^j) > U_j(p^k)$.
\end{defn}
    

\begin{defn}[Envy-free allocation]
    A random allocation $p$ is \textbf{envy-free} if no agent envies another agent.
\end{defn}

\section{Existence result}\label{sec2}

Using theorem \ref{thm:wpe}, to characterize the set of weakly Pareto efficient allocations, define a \textbf{primal problem} $\mathbf{P(p,\lambda,U)}$ for a fixed set of utility representations and a vector of nonnegative weights $(\lambda_j)_{j \in J}$

\begin{equation}
    P(p,\lambda,U) = \sum_{j\in J} \lambda_j U_j(p)
\end{equation}

As proved in the previous section, the set of weak Pareto allocations is equal to the set of argmaximizers of the program above as the vector of weights $\lambda$ varies over the simplex.

To establish the existence of PE and EF allocation, I am adopting the approach of Echenique et al \cite{EcheniqueEtAl}. The proof relies on the existence of a converging subnet of argmaximizers. Since lower-hemicontinuity of demand correspondence is impossible to ensure in general problems, the first step in the proof is to relax weak Pareto efficiency to $\epsilon$-Pareto efficiency and introduce a new strictly concave objective function that ensures the uniqueness of a maximizer, and henceforth, continuity of the argmax selection.

For a vector of nonnegative weights $(\lambda_j)_{j \in J}$, and a strictly positive real number $\delta$, define a modified primal maximization problem $\bm{Q(p,\lambda,U,\delta)}$
\[
    Q(p, \lambda,U,\delta) = \sum_{j \in J} \lambda_j U^j(p) - \delta G(p)
\]

For the rest of the text, for the simplicity of exposition, the dependence of $Q$ on the utilities will be omitted. The reader should keep in mind that all the results hold for a fixed family of $\{U_j\}_{j \in J}$.


To ensure uniqueness, the concave objective is modified by adding a continuous, strictly concave map $-G: \Delta(X) \rightarrow \R$. The existence of a strictly convex map in metric spaces is a non-trivial exercise. Moreover, due to the nature of a crucial lemma, I also require this map to be invariant under permutations, that is, $G(p_{j \xleftrightarrow[]{} k}) = G(p)$ for any $j,k$. The next paragraph proves the existence of such a map.

Let $\Delta(\pi_j(X))$ be the space of all probability measures on $\pi_j(X)$, a compact subspace of $Y$, hence a compact metric space. Note that the index of agent $j$ is redundant since $\pi_j(X) = \pi_k(X) \ \forall k,j \in J$.

Then, $\Delta(\pi_j(X))$ is compact and metrizable. Using the result due to Herve \cite{Alfsen1971}, a compact, convex set $C$ admits a strictly convex, continuous, real-valued function iff $C$ is metrizable. Let $F:\Delta(\pi_j(X)) \rightarrow \R$ be the map.

Define $G(p) = \sum_{j \in J} F(p^j)$. The composition of a strictly increasing convex and strictly convex function is strictly convex; the easy proof of this result is omitted. Additionally, $G$ satisfies invariance under permutations $G(p) = G(p_{i \leftrightarrow j}) \ \forall i,j \in J$.

\subsection{Existence of an envy-free agent}\label{subsec2}

I start by defining the normalized space of Pareto weights, the simplex $\Delta^{n-1} = \{\lambda \in R^n_+: \sum_j \lambda_j = 1\}$.

\begin{defn}
    The \textbf{support} (carrier) of $\lambda \in \Delta^{n-1}$ is $\chi(\lambda) = \{j: \lambda_j > 0\}$.
\end{defn}

\begin{lem}[Lucky] \label{lem:lucky}
    Fix $\lambda \in \Delta^{n-1}$. Let $p(\lambda, \delta)$ be an allocation that solves $Q(\lambda, \delta)$. There exists agent $i \in \chi(\lambda)$ who is EF at $p(\lambda, \delta)$.
\end{lem} 

The original lemma is due to Varian \cite{Varian1974}. Here I present a version of the proof for the setting of this paper.

\begin{proof}
    By contradiction. Assume that every agent in support envies someone.
    First, observe that an agent in support cannot envy an agent not in support. Toward contradiction, suppose agent $i$ with $\lambda_i > 0$ envies agent $j$ and $\lambda_j = 0$, that is $U_i(p^j) > U_i(p^i)$. Define a new allocation $p_{i \rightarrow j}$. Since $G(p) = G(p_{i \rightarrow j})$, $Q(p_{i \rightarrow j},\lambda,\delta) > Q(p,\lambda,\delta)$.
    Therefore, every agent belongs to a cycle of envy involving only agents in the support. Denote the cycle of length $k$ by $\{j_1, j_2, ..., j_k\}$. Define an after-exchange allocation by $p_{j_1 \xleftrightarrow[]{} \dots \xleftrightarrow[]{} j_k} = p \circ (T_{j_k \xleftrightarrow[]{} \dots \xleftrightarrow[]{} j_1})^{-1} = p \circ T^{-1}_{j_2 \xleftrightarrow[]{}j_{1}} \dots T^{-1}_{j_k \xleftrightarrow[]{}j_{k-1}}$. Since $T$ is a linear operator, it follows the laws of composition and inversion of linear operators.
    
    Since for every two agents $j_l, j_{l+1}$, agent $j_l$ envies agent $j_{l+1}$, $U^{j_l}(p) < U^{j_l}(p_{j_{l} \xleftrightarrow[]{} j_{l + 1}})$. Since agent's utility depends only on their marginal probability distribution, $U^{j_l}(p_{j_{l} \xleftrightarrow[]{} j_{l + 1}}) = U^{j_l}(p_{j_1 \xleftrightarrow[]{} \dots \xleftrightarrow[]{} j_k})$.
    
    In the objective
    \[
        \sum_{j \in J} \lambda_j U^j(p_{}) - \delta G(p)
    \]
    the weighted sum strictly increases, and since $G$ is invariant under permutations, $G(p) = G(p_{j_1 \xleftrightarrow[]{} \dots \xleftrightarrow[]{} j_k})$. Therefore, interchanging their allocations strictly increases the value of $Q(p,\lambda, \delta)$.
\end{proof}

\subsection{Simplicial subdivision}\label{subsec3}

By \ref{thm:wpe}, any allocation that maximizes the weighted sum of agents' utilities, namely $p \in \arg\max_{p \in \Delta(X)}P(p,\lambda,U)$, is wPE. The main idea of the proof is to find the "right" vector of weights, such that for the vector of these weights, there is an allocation $p^*(U)$ that is envy-free. 
In search of the envy-free vector of weights, I will employ simplicial subdivision of the simplex of weights. I choose barycentric subdivision for its well-studied properties.

For the existence theorem of this paper, the reader only needs to know two things about the subdivision:
\begin{enumerate}
    \item Barycentric subdivision of a simplex produces a collection of subsimplixes that cover the original simplex.
    \item Repeated application of the barycentric subdivision can produce a collection of subsimplexes of arbitrary fine mesh: for every $\epsilon > 0$, the diameter of the largest subsimplex can be taken to be smaller than $\epsilon$. 
\end{enumerate}

For curious readers, I briefly describe the barycentric subdivision procedure following Kim Border's book \cite{Border1985}.

Let $\Delta^{n-1} = \{(\lambda_1,\dots,\lambda_n) \in \mathbb{R}^n \Big| \lambda_j \geq 0, \sum_{j \in J}\lambda_j = 1\}$ be a standard closed $n$-simplex. It is the closed convex hull of $\{e_1,\dots,e_n\}$, where $e_j$ is the $j$th vector of the standard basis of $\mathbb{R}^n$. More generally, define a $n-1$ simplex $T$ by the collection of its vertices $T = [x_1,\dots,x_n]$:
\[
    [x_1,\dots,x_n] = \{\sum_{j = 1}^n \lambda_j x_j \ \Big| \ \lambda_j \geq 0, \sum_{j = 1}^n \lambda_j = 1\}
\]
\begin{defn}
    The \textbf{barycenter} of a general simplex $T$ is the point $\frac{1}{n}\sum_{j = 1}^n x_j$. The barycenter is denoted by $b(T)$. For example, for $\Delta^{n-1}$, $b(\Delta^{n-1})$ is the point $\frac{1}{n}\sum_{j = 1}^n e_j$.
\end{defn}

Define the partial order on the set of simplises as follows: $T_1 > T_2$ if $T_1 \neq T_2$ and $T_2$ is a face of $T_1$.

\begin{defn}
    Given a simplex $T$, \textbf{barycentric subdivision} of $T$ is a collection of simplices $[b(T_1),\dots,b(T_k)]$ such that $T \geq T_1 > \dots > T_k$.
\end{defn}

Finally, I use the important property of the barycentric subdivision. Let $\Delta \subset \R^n$ be a simplex. Define $\text{diam}(\Delta) = \max \norm{a-b}_{\R^n}$ such that $a,b \in \Delta$. Let $\Delta^\prime$ be an n-dimensional simplex that comes from the covering of $\Delta$ obtained by the barycentric subdivision. Then, the following estimation holds:
\[
    \text{diam}(\Delta^\prime) \leq \frac{n}{n+1}\text{diam}(\Delta)
\]

Therefore, by applying barycentric subdivision sufficiently often, the largest diameter can be made as small as desired.

Simplicially subdivide the original simplex of weights $\Delta^{n-1}$ using barycentric division. Let $V$ denote the collection of vertices of all the subsimplixes. A function $l:V \rightarrow \{0,1,...n\}$ is called proper labeling if 
\[
    l(v) \in \chi(v)
\]

where $\chi(v)$ is a support (carrier) of $v$ (observe that every $v$ is a convex combination of the vertices of the original $\Delta^{n-1}$ simplex).

The following corollary is used to define a proper labeling of any barycentric subdivision. It follows immediately from the Lucky lemma \ref{lem:lucky}.

\begin{cor}
    The labeling function that assigns to the vertex $v$ the index $i$ of the agent who is envy-free at $v$, is a proper labeling.
\end{cor}

\begin{thm}[Sperner's lemma]
        Let simplex $\Delta^{n-1}$ be simplicially subdivided and properly labeled. Then there are an odd number of completely labeled simplices in the subdivision.
\end{thm}
    
Sperner's lemma states that there is a completely labeled subsimplex $\Delta^{n-1}_k$, that is $l(V_{\Delta^{n-1}_k}) = \{1,...n\}$. In other words, the labeling map takes all the values on the vertices of this subsimplex.

\subsubsection*{Useful correspondence theorem}

\begin{thm}[Berge's maximum theorem] \label{thm:max}
    Let $X$ and $\Theta$ be topological spaces. Let $f: X \times \Theta \rightarrow \R$ be continuous on the product space $X \times \Theta$. Let $C:\Theta \rightarrow \rightarrow X$ be a non-empty compact-valued correspondence.
    Define the argmax correspondence by
    \[
        C^*(\theta) = \{x \in X: f(x,\theta) \geq f(\tilde{x},\theta) \ \forall \tilde{x} \in C(\theta)\}
    \]
    If $C$ is continuous at $\theta$, then $C^*(\theta)$ is upper hemicontinuous at $\theta$ and non-empty compact-valued.
    Note: a set-valued map that is singleton-valued is upper-hemicontinous if and only if the corresponding function is continuous.
\end{thm}

\subsection{Existence theorem}\label{subsec4}

\begin{thm} \label{thm:existence}
    A weak Pareto-efficient, envy-free allocation exists.
\end{thm}

The proof of the theorem is two-fold. The majority of the effort will be spent proving the following lemma.

\begin{lem}
    For every family of utility functions $U = \{U_j\}_{j \in J}$, for every $\delta > 0$, there exists $\bar{\lambda}$ such that $p = \arg \max_{p \in \Delta(X)} Q(p,\bar{\lambda}, U, \delta)$ is envy-free. $\bar{\lambda}$ is independent of $\delta$.
\end{lem}

\begin{proof}
    Fix a sequence of real numbers $c_k \downarrow 0$. For every $k \in \N$, simplicially subdivide the simplex $\Delta^{n-1}$ (using barycentric subdivision, for example), such that diameter of every subsimplex is less than $c_k$. For every $k$, use Sperner's lemma to pick a completely labeled subsimplex $\Delta^{n-1}_k$ and choose a sequence $\lambda_k$ such that for each $k$, $\lambda_k$ is an element of a completely labeled subsimplex $\Delta^{n-1}_k$ and $diam(\Delta^{n-1}_k) < c_k$. By Bolzano-Weierstrass, $\lambda_{k_l} \rightarrow \bar{\lambda}$. I claim that $p^*(\bar{\lambda}, \delta) = \arg \max Q(p,\bar{\lambda}, U, \delta)$ is EF.

    Toward contradiction, assume that $p^*(\bar{\lambda}, \delta)$ is not envy-free. Then there is an agent pair $(i,j)$ such that $U_i(p^*(\bar{\lambda},\delta)) < U_i(p^*_{i \xleftrightarrow[]{}j}(\bar{\lambda},\delta))$. The exchange operator $L_{i \xleftrightarrow[]{}j}(p) = p_{i \xleftrightarrow[]{}j}$ is continuous in weak-* topology, and Berge's Maximum Theorem \ref{thm:max} guarantees the continuity of $p^*(\cdot,\delta)$, so there is $\epsilon > 0$ such that $\forall \lambda_k \ : \ \norm{\lambda_k - \bar{\lambda}} < \epsilon$, $U_i(p^*(\lambda_k,\delta)) < U_i(p^*_{i \xleftrightarrow[]{}j}(\lambda_k,\delta))$. 
    Since $\lambda_{k_l} \rightarrow \bar{\lambda}$, there exists $\bar{k_l}$ such that for all $k_l \geq \bar{k_l}$, $d(\lambda_{k_l}, \bar{\lambda}) < \epsilon/2$, Since $c_k \downarrow 0$, take a $\tilde{k_l}$ such that $\tilde{k_l} \geq \bar{k_l}$ and $c_{\tilde{k_l}} < \epsilon/2$. Then for every $\lambda \in \Delta^{n-1}_{\tilde{k_l}}$, $d(\lambda, \bar{\lambda}) \leq d(\lambda,\lambda_{\tilde{k_l}}) + d(\lambda_{\tilde{k_l}}, \bar{\lambda}) < \epsilon$, therefore $\Delta^{n-1}_{\tilde{k_l}}$ cannot be completely labeled, which is a contradiction.

\end{proof}

Note that the continuity of the exchange operator $L_{i \xleftrightarrow[]{}j}: \Delta(X) \rightarrow \Delta(X)$ follows from a more general result.

\begin{lem}
    Let $X$ be a Polish space, and $T: X \rightarrow X$ a continuous operator. Define a push-forward operator $L_T: \Delta(X) \rightarrow \Delta(X)$ by $L_T(p) = p \circ T^{-1}$. Then $L_T$ is continuous in weak-* topology.
\end{lem}

The proof of the lemma is in the Appendix.

The rest of the proof of the theorem ensures the EF property as $\delta \rightarrow 0$.

\begin{enumerate}
    \item Fix $\epsilon > 0$. Choose $\delta(\epsilon)$ such that 
    \[
        \max_p \delta(\epsilon) G(p) \leq \epsilon
    \]
    
    \item Lemma above asserts that there exists $\bar{\lambda}$ such that $p \in \arg \max Q(\bar{\lambda}, \delta(\epsilon))$ is EF.
    
    \item Take a sequence of $\epsilon_n \rightarrow 0$. Observe $\delta_n = \delta(\epsilon_n) \rightarrow 0$. Take the corresponding sequence of $p(\bar{\lambda},\delta_n)$.

    Since $X$ is metrizable and compact, $\Delta(X)$ is a compact metric space, and for metric spaces, compactness and sequential compactness are equivalent. Since $p(\bar{\lambda},\delta_n) \in \Delta(X)$, it has a convergent subsequence $p(\bar{\lambda},\delta_{n_k})$.
    Denote its limit by $p^*(\bar{\lambda})$. If this allocation is not EF, there is agent $i$ who envies another agent $j$. Replicating the argument above, using continuity and strict inequality that characterizes envy-freeness, I conclude that there is $p(\bar{\lambda},\delta_{n_l})$ which is also not EF.
\end{enumerate}

\section{Applications} 

\subsection{Allocating indivisible goods/services over time}

Long-horizon planning is a fundamental component of many economic environments with dynamic incentives and intertemporal trade-offs. At the same time, real-world allocations are frequently subject to indivisibilities at each point in time, while institutional and normative considerations attempt to ensure fairness and equitable treatment.

Consider the industry of commercial aviation. All over the airports of the world, airlines compete strategically for a number of indivisible items, like take-off/landing time slots, baggage capacity, maintenance stands, even de-icing trucks in winter. Moreover, the nature of any airline decision-making process is dynamic, since flight timetables are designed years in advance.

Airlines preferences are aligned but not identical. They are determined by the airline’s business model. For example, low-cost operators prefer simple routes and small baggage capacities, while full-service airlines offer convoluted connections and extra luggage space. Long-haul premium compete with elevated experiences and comfortable departure/arrival times.

On the other side of things, there are airport operators that are natural monopolies and are already institutionally regulated. For example, at an international airport in Europe, time slot assignment is handled by an independent coordinator under rules aligned with International Air Transport Association guidelines and local law. The rules highlight fairness. The regulation No 95/93 (the Slot Regulation), issued by the Council of European Union, was adopted to ensure that airlines have access to the busiest EU airports "on the basis of neutrality, transparency and non-discrimination".

The primitives of the allocation problem of the airline industry are as follows. There are $m \geq 1$ distinct indivisible services (e.g.\ $m$ different facilities or amenities). They are jointly allocated at each point in time $n \in \N$.

\begin{defn}
    For $m\in\N$, the \textbf{$m$-service Cantor space} is
    \[
      \mathcal{C}^m \;=\; \{0,1\}^m \times \{0,1\}^m \times \cdots
              \;=\; \bigl(\{0,1\}^m\bigr)^{\!\infty}
              \;\cong\; \{0,1\}^{m\times\infty},
    \]
    endowed with the product topology. 
\end{defn}

To make things more concrete, let's consider a simplified version of the airline allocation problem, where the items that are allocated are time slots for take-off and landing. Let $n \in \N$ denote one 15-min time slot for take-off/landing at any of the $m \geq 1$ international airports around the world.

An element $A = (A_n)_{n\geq 0}\in \mathcal{C}^m$, where
$A_n = (a_n^1,\ldots,a_n^m)\in\{0,1\}^m$, records for each period $n$ the allocation
vector of all $m$ services:
\[
  a_n^l \;=\;
  \begin{cases}
    1 & \text{time slot number $n$ }\text{ is given to an airline at the airport } l,\\
    0 & \text{time slot number $n$} \text{ is not given to the airline at the airport } l.
  \end{cases}
\]

Consumption space of each airline is $Y \subset \mathcal{C}^m$, incorporating the constraints that an airline faces, such as:
    \begin{enumerate}
        \item Fleet capacity constraint: for every $n$, $\sum_{l = 1}^m a^l_n(j) \leq U$
        \item Minimum activity level: for every airport, for a fixed time interval $T$ (representing a day), $L \leq \sum_{n = t}^{t + T} a^l_n(j)$ for every $t = 0, T, 2T, \dots$
        \item Inventory balance constraint: For each airport $l$ and time slot $n$, define:
            \[
            a_{n}^{l,\mathrm{arr}} \in \{0,1\}
            \]
        to represent arrivals, and
            \[
            a_{n}^{l,\mathrm{dep}} \in \{0,1\}
            \]
        to represent departures. Denote aircraft stock via at the airport $l$ at time $n$ by $S_n^l$.
        It captures the number of airline's aircraft physically present at airport $l$ immediately after slot $n$. The stock-flow equation is
            \[
            S_{n+1}^{l}
            = S_{n}^{l} + a_{n}^{l,\mathrm{arr}} - a_{n}^{l,\mathrm{dep}}
            \] 
        Starting with the initial condition $S_n^l$, the fleet ready for the departure at no point in time can exceed the fleet that arrived by the time $n$, that is for every $n \in \N$, for every airport $l \in m, S_n^l \geq 0$.
    \end{enumerate}

Let $\nu: \N \rightarrow \N^m$ be a total endowment of the economy that captures how many of each services $l = 1, \ldots, m$ are available every period. The space of deterministic allocations is cut out of Cartesian product by a standard feasibility constraint:
    \[
        X = \{(A(1), \dots, A(N)) \ \Big| \ A(j) \in Y, \sum_{j \in J} A(j) \leq \nu\}
    \]

Using the existence result \ref{thm:existence}, I claim that for a rich class of preferences, weak Pareto efficient and Envy-free allocation of the time slots to the airlines over the international airports of the world exists. For the result to be applicable, I show that $X$ is a compact metrizable space. \footnote{I am grateful to 
Fabio Maccheroni for his valuable input on the application.}

Let me introduce the discounted metric on $\mathcal{C}^m$, which is defined as follows:
    \[
      d_\beta^m(A,B)
      \;=\; (1-\beta)\sum_{n=0}^{\infty}\beta^n\,
            \frac{1}{m}\sum_{l=1}^{m}|a_n^l - b_n^l|
      \qquad A,B\in \mathcal{C}^m.
    \]

\begin{prop}
    $\mathcal{C}^m$ is metrizable (discounted metric) compact space.
\end{prop}

\subsection{Allocating fairly and efficiently in economies with differentiated goods}

Consider a problem of allocating heterogeneous, indivisible delivery orders to drivers in a stochastic platform environment.

For example, Amazon delivery drivers face a population of delivery opportunities (orders) indexed by characteristics such as weight, volume, and location. These orders are indivisible and can only be consumed in integer units, while drivers are subject to a form of stochastic access friction arising from the matching process: available orders are revealed through a dynamic platform interface, and drivers may miss opportunities due to timing uncertainty, asynchronous updates, or algorithmic assignment variability.

Drivers' preferences are rich enough to make the question of existence of a fair and efficient matching non-trivial. The set of “lucrative” orders — those with high value across drivers — is limited in supply, generating congestion-like competition for a small number of highly ranked opportunities. At the same time, drivers exhibit heterogeneous preferences over non-price attributes of orders: some prefer remote destinations, others urban routes; some favor fewer heavy packages, while others prefer larger bundles of lighter deliveries; and preferences also vary over temporal dimensions such as morning versus evening shifts.

This combination of indivisible goods, stochastic allocation opportunities, and multidimensional preference heterogeneity naturally lends itself to an economy with differentiated commodities in the spirit of Mas-Colell (1975) \cite{MasColell1975}, where orders can be interpreted as differentiated goods.

MasCollel describes an exchange economy in which differentiated commodities are available in integral amounts (that is, not perfectly divisible) and the variety of commodities comes from the richness of characteristics. Let $(K,d)$ be a compact metric space, called the space of commodity characteristics. The space of individual commodity bundles is the space of non-negative, bounded Borel measures that are integer valued. The consumption set $Y$ is the set of all individual commodity bundles such that $a(K) \leq \alpha$, where $\alpha$ is a large positive integer. 

Note that a set of all integer-valued Borel measures on a compact metric space $X$, which are uniformly bounded by an integer $\alpha$, is compact in weak-* topology (proof is in the Appendix).

Let $\nu$ be a total endowment measure of the economy such that $supp(\nu) = K$. The space of deterministic allocations is cut out of Cartesian product by a standard feasibility constraint:
\[
    X = \{(a_1, \dots, a_n) \ \Big| \ a_j \in Y, \sum_{j \in J} a_j \leq \nu\}
\]

Using the existence result \ref{thm:existence}, I claim there is a way to allocate orders efficiently and fairly. I would like to highlight that perceived fairness is an important concern for platform-based labor markets because workers who feel they are treated unfairly may reduce participation, try to game the system, or stop trusting the platform altogether.

For example, this issue has appeared repeatedly in the context of Amazon Flex, where drivers have complained that the process determining who receives the most attractive delivery blocks is opaque and difficult to understand. In response, the platform introduced driver rating systems, increased enforcement against bots used to capture desirable routes, and made parts of the allocation process more structured. Still, many drivers continue to express concerns about transparency and about whether delivery opportunities are allocated fairly across workers.

\subsection{Existence of wPE and EF allocation in the cake cutting problem}

The problem of allocating land or, more broadly, allocating a divisible good among agents in a way that satisfies certain fairness and/or efficiency criteria traces back to Steinhaus \cite{Steinhaus1948}, who formalized the problem known as the cake-cutting problem, starting the literature on fair division.
In this literature, allocations are modeled as partitions of unity.

\begin{defn}
    A \textbf{partition of unity} on a topological space \( X \) is a finite set \( R \) of functions from \( X \) to the unit interval \([0,1]\) such that for every point \( x \in X \) the sum of all the function values at \( x \) is 1, i.e.,
        \[
        \sum_{\rho \in R} \rho(x) = 1.
        \]
\end{defn}

Preferences for land (or cake) are modeled as finite measures, and whether it is assumed that the measures are atomless or not calls for conceptually different interpretations of allocations. 

The modeling of allocations as partitions of unity is justified as an intermediate technical step. When agents preferences are assumed to be atomless, it is shown that fair allocation can be taken to be a partition.

Regardless of whether utility measures are atomless or not, applying the main theorem to the partitions of unity instead of marginal probability measures is formally justified. In the following paragraphs, I show that for every probability measure on a deterministic consumption space of land partitions, there exists a payoff-equivalent partition of unity.

To ensure that the existence theorem applies, I need to make stronger assumptions on the space that is to be divided. Instead of working with $(L, \mathcal{B})$, a measurable space, I assume that $L$ is a separable first countable Hausdorff topological space together with its Borel sigma-algebra $\mathcal{B}$. The set of deterministic consumption $Y$ is identical for all agents. It is the set of all measurable subsets of $L$, $Y = \mathcal{B}$. The set of deterministic allocations $X \subset Y^N$ contains partitions of $L$ into $N$ subsets: $X = \{(F_1, \dots, F_N) \ \Big | \ F_i \cap F_j = \emptyset \ \ \forall i,j \in N; \ \bigcup_{i \in N} F_i = L \}$.

In order to talk about probability measures on $X$, start with $(Y,\tau^\prime)$ as a topological space. I endow it with a relative topology from the power set $2^L$. This topology is constructed by extending the Vietoris (or Fell) topology for closed subsets of the space (see the appendix for the extension, \cite{Hildenbrand1974}). Denote by $\mathcal{M}$ the set of Borel sets of $Y$, and take $(y,\mathcal{M})$ to be a measure space.

$(X,\tau)$ is a topological space with a subset topology inherited from a product topology of $Y^N$. Denote by $(X,\mathcal{B}^\prime)$ a measure space $X$ together with its Borel sigma-algebra.

A probability measure $p$ on the set of deterministic allocations is a mapping from $\mathcal{B}^\prime$ to $[0,1]$ that satisfies the countable additivity property. It generates a family of marginal probability measures $\{p_j\}_{j \in \{1, \dots,N\}}$, each on the space of all measurable subsets of $L$.

\begin{lem} \label{lem:payoff}
    Let $P$ be a probability measure on the space of all measurable subsets of $L$, $P \in \Delta(\mathcal{B})$. There exists $f: L \rightarrow [0,1]$ such that for all finite measures $\mu$ on $L$
    \[
        \int_\mathcal{B} \mu(F) dP(F) = \int_L f(x) d \mu(x)
    \]
\end{lem}

\begin{proof}
    Observe that $(L,\mathcal{B}, \mu)$ and $(Y,\mathcal{M},P)$ are $\sigma$-finite measure spaces. Define $g: L \times \mathcal{B} \rightarrow [0,1]$ as
    \[
        g(y,F) = 
        \begin{cases}
        1, & \text{if } y \in F \\
        0, & \text{otherwise}
        \end{cases}
    \]
    Note that $g$ is nonnegative and $\mathcal{B} \otimes \mathcal{M}$ - measurable (proof in the appendix). Note that for every $F$, $\mu(F) = \int_L g(y,F) d\mu(y)$. By Tonelli's theorem,
    \[
        \int_\mathcal{B} \mu(F) dP(F) = \int_\mathcal{B} \int_L g(y,F) d\mu(y) dP(F) = \int_L \int_\mathcal{B} g(y,F) dP(F) d\mu(y) = \int_L f(y) d \mu(y)
    \]
    where $f(y) = \int_\mathcal{B} g(y,F) dP(F) = P(\{F \in \mathcal{B}:y \in F\})$. By construction $f(L) \subset [0,1]$.
\end{proof}

Therefore, for an agent with preferences $\mu_j$, for every random allocation, there is a bounded map $f_j: L \rightarrow [0,1]$ that gives them the same payoff.

Let $P$ be a probability measure on $X$. It gives rise to a family of marginal probability measures $\{P_j\}_{j \in J}$.

\begin{lem}
    Let $P \in \Delta(X)$. Then $\sum_{j \in J} f_j = 1$, where $f_j$ is a mapping from Lemma \ref{lem:payoff} corresponding to $P_j$.
\end{lem}

\begin{proof}


    For $x \in L$, let $A_j(x) = \{(F_1,\dots,F_j,\dots,F_N): x \in F_j \ \text{and} \ F_k \cap F_l = \emptyset \ \forall k,l \in J; \bigcup_{k \in J} F_k = L\}$. $A_j(x)$ denotes the set of all partitions, in which point $x$ is allocated to agent $j$. For every $x \in L$, note that the sets $\{A_j(x)\}_{j \in J}$ are pairwise disjoint, that is, $A_i(x) \cap A_j(x) = \emptyset \ \forall i \neq j$: if $(\tilde{F}_1, \dots, \tilde{F}_N) \in A_j$, then $x \in \tilde{F}_j$ and $x \notin \tilde{F}_i$, hence $(\tilde{F}_1, \dots, \tilde{F}_N) \notin A_i$.
    Also for every $x$, $\cup_{j \in J} A_j = X$. The direction $\cup_{j \in J} A_j \subset X$ is trivial. For the other direction, take $(\tilde{F}_1, \dots, \tilde{F}_N) \in X$. Since it is a partition, let $j$ denote the set such that $x \in \tilde{F}_j$. Then $(\tilde{F}_1, \dots, \tilde{F}_N) \in A_j$.

    Hence, for every $x$, $1 = P(X) = P(\bigcup_{j \in J} A_j) = \sum_{j \in J} P(A_j) = \sum_{j \in J} P_j(\{F: x \in F\}) = \sum_{j \in J} f_j(x)$.
\end{proof}


From now on, the set of stochastic allocations is taken to be the set of partitions of unity on $L$, instead of the set of probability measures on the set of deterministic allocations. I denote the set of stochastic allocations by $Z$:
\[
    Z = \{(f_1, \dots, f_N) \ \text{such that for every} \ j, f_j: L \rightarrow [0,1] \ \text{and} \ \sum_{j \in J} f_j = 1 \}.
\]

Following the cake-cutting literature, I embed partitions of unity into $L^\infty(L, \lambda)$ for some measure $\lambda$. Usually it is defined as an "average taste" measure $\lambda = \sum_{j \in J} m_j \mu_j$ for a collection of $\{m_j\} \in \mathbb{R}^N$ such that $\sum_{j \in J} m_j = 1, m_j \geq 0$ (agents' weights).

In particular, Berliant \cite{Berliant1985} takes $L$ to be a compact subset of $\mathbb{R}^2$. He models preferences as measures that are absolutely continuous with respect to the Lebesgue measure on $\mathbb{R}^2$. For agent $j$, the utility from an allocation $(f_1, \dots, f_N)$ is equal to $U_j(f_j) = \int_{L} f_j u_j d\lambda$, where $\lambda$ is the Lebesgue measure and $u_j$ is a corresponding Radon-Nikodym derivative.

The topology on $L^\infty(L, \lambda)$ is a weak-* topology. Observe that the set of allocations $Z$ is closed in the product topology:
\[
    Z = \bigcap_{j \in 1,\dots,N} \{(f_1, \dots, f_N) : f_j \geq 0\} \bigcap \{(f_1, \dots, f_N) : \sum_{j \in 1,\dots,N} f_j = 1\}
\]
For every $j$, $\{(f_1, \dots, f_N) : f_j \geq 0\} = \bigcap _{x \in L}\{(f_1, \dots, f_N) : f_j(x) \geq 0\}$, where for every $x \in L$, $\{(f_1, \dots, f_N) : f_j(x) \geq 0\}$ is the preimage of a closed set $[0, \infty)$ under the composition of a linear functional $f_j \rightarrow f_j(x)$  with the projection mapping $\pi_j$, both continuous by construction. 
Analogously, $\{(f_1, \dots, f_N) : \sum_{j \in 1,\dots,N} f_j = 1\} = \bigcap_{x \in L} \{(f_1, \dots, f_N) : \sum_{j \in 1,\dots,N} f_j(x) = 1\}$, and for every $x$, $\{(f_1, \dots, f_N) : \sum_{j \in 1,\dots,N} f_j(x) = 1\}$ is the preimage of a closed set under the composition $(f_1, \dots, f_N) \in (L^\infty)^N \rightarrow \sum f_j \in L^\infty \rightarrow \sum f_j (x)$. The summation map is continuous since $L^\infty$ is a TVS.

The last bit is to show that the map $f = (f_1, \dots, f_N) \rightarrow (U_1(f), \dots U_N(f))$ is continuous in the weak-* topology of $L^\infty (L, \lambda)$. The proof is due to Akin \cite{Akin1995}. Observe that $U_j: f_j \rightarrow \int f_j d \mu_j$ is continuous in weak-* topology of $L^\infty(L, \mu_j)$ since $U_j(f_j) = \int f_j g d\mu_j $ with $g \in L^1(\mu_j), g = 1$. Given that $\mu_j$ is absolutely continuous wrt $\mu$, $U_j(f_j) = \int f_j g d\mu $ with $g \in L^1(\mu), g = \frac{d\mu}{d\nu}$.


Since $L$ is separable, $L^1(L,\lambda)$ is a separable Banach space, so the unit ball of its dual space $L^\infty(L,\lambda)$ is metrizable and compact (Banach - Alaoglu). $Z$ is a compact metric space, which is also convex and satisfies the anonymity property. The set of its so-called marginals $f_j$ is identical for all agents. 

Note that these properties of the set of allocations, together with the agents' preferences represented as finite positive measures, are sufficient to apply the lucky lemma \ref{lem:lucky} and the existence theorem \ref{thm:existence} in the same fashion as to the set of random allocations $\Delta(X)$, constructed in the setup. By applying the existence theorem to the set $Z$, I conclude that there exists a weakly Pareto efficient and envy-free allocation of land $\bar{f}$ among $N$ agents. Notice that the existence result is independent of whether preference measures are atomless or not.

Moreover, in both cases, the allocation in question can be characterized in greater detail. To this end, I invoke several results from Dvoretzky, Wald, and Wolfowitz \cite{DvoretzkyWaldWolfowitz1951}, later referred to as DWW.

Before I state the theorems, let me introduce some definitions. Given an allocation $f \in Z$, let $\nu(f;\mu)$ be the $N \times N$ dimensional vector of agents' utilities from their allocation and the allocation of every other agent
$$\nu(f) = \begin{bmatrix}
    \int_L f_1(x) d\mu_1(x), \dots, \int_L f_1(x) d\mu_N(x), \\
    \dots \\
    \int_L f_N(x) d\mu_1(x), \dots, \int_L f_N(x) d\mu_N(x)
\end{bmatrix}$$

The set of agents' utilities under all feasible stochastic allocations (partitions of unity) is denoted by $V_N(\mu_1, \dots, \mu_N)$. It is the set of all vactors $\nu(f;\mu)$, where $f \in Z$. 

The set of agents' utilities under all feasible stochastic allocations that are simple functions is denoted by $V_N^0(\mu_1, \dots, \mu_N)$. It is the set of all vectors $\nu(f^0;\mu)$, where $f^0 = (f_1^0,...,f_N^0)$ such that $f_j^0 = \sum_{k} a_k\chi_{E_k}$ and $\sum_{j \in J}f_j = 1$. In other words, $f_j$ is a simple (step) function for every $j$.

Finally, the set of agents' utilities under all feasible deterministic allocations is denoted by $V_N^*(\mu_1, \dots, \mu_N)$. It is the set of all points $\nu(f^*;\mu)$, where $f^* = (f_1^*,...,f_N^*)$ such that $f_j^* = \chi_{E_j}$ and $\cup_{j \in N} E_j = L$. 

\begin{thm}[Theorem 3 in DDW] \label{thm:3}
    If measures $\mu_1, \dots, \mu_N$ are finite, the simple range $V^0_N(\mu)$ coincides with $V_N(\mu)$: every point in $V_N$ can be represented as $\nu(f^0)$ for some $f^0$. 
\end{thm}

\begin{thm}[Theorem 4 in DDW] \label{thm:4}
    If measures $\mu_1, \dots, \mu_N$ are finite and atomless, the degenerate range $V^*_N(\mu)$ coincides with the simple range $V^0_N(\mu)$: every point in $V_N^0$ can be represented as $\nu(f^*)$ for some $f^*$. 
\end{thm}

Observe that \ref{thm:3} and \ref{thm:4} imply that the sets $V_N$ and $V_N^*$ are equal if $\mu_1, \dots, \mu_N$ are finite atomless measures. Recall that $\bar{f} \in Z$, hence $\nu(\bar{f}) \in V_N$. There exists a vector of indicator functions that constitute a partition ($f = (f_1,...,f_N) \in Z$ such that $f_j = \chi_{E_j}$ and $\cup_{j \in N} E_j = L$), that gives agents the same utilities as $\bar{g}$. Therefore, a feasible deterministic allocation $\chi_{E_1}, \dots, \chi_{E_N}$ is wPE and EF.

Now, let's consider the case when agents' measures $\{\mu_1,\dots, \mu_N\}$ are not atomless. As in the case of allocating finitely many items, there is little hope that an envy-free allocation is deterministic. However, I can show that wPE and EF allocation can be taken to have finite support. Using \ref{thm:3} together with the existence theorem \ref{thm:existence}, I know that wPE and EF allocation can be taken to be a family of simple functions. In the next theorem, in the spirit of the Birkhoff-von-Neumann theorem, I show that every such family of simple functions can be represented as a convex combination of families of indicator functions.

\begin{thm}
    Let $\{f^0_j\}_{j \in J}$ be a family of simple functions that is a partition of unity. Then there exist finitely many families of indicator functions $\{f^*_{k,j}\}_{k \in K,j \in J}$ such that $f^*_{k,j} \geq 0, \sum_{j \in J}f^*_{k,j} = 1$ for every $k$ and such that $(f_1^0, \dots, f_N^0)$ is in the convex hull of $\{(f_{1,1}^*, \dots f_{1,N}^*), \dots, (f_{K,1}^*, \dots f_{K,N}^*)\}$:
    \[ 
        \begin{bmatrix}
        f_1^0 \\
        f_2^0 \\
        \vdots \\
        f_N^0
        \end{bmatrix}
        =
        \sum_{k=1}^{K} a_k
        \begin{bmatrix}
        f_{1,k}^* \\
        f_{2,k}^* \\
        \vdots \\
        f_{N,k}^*
        \end{bmatrix}
    \]
    where $a_k \geq 0, \sum_{k \in K} a_k = 1$.
\end{thm}

\begin{proof}
    Here, I present a nice and intuitive proof for the case of two agents. WLOG, the two simple functions can be defined on the same partition $\{E_i\}_{i \in n}$ of $L$. Let $f^0_1 = \sum_{i = 1}^n \alpha_i \chi_{E_i}$ and $f^0_2 = \sum_{i = 1}^n \beta_i \chi_{E_i}$ such that $\alpha_ i + \beta_i = 1$. Order the coefficients for agent 1 in a non-decreasing order $\alpha_1 \leq \alpha_2 \leq \dots \leq \alpha_n$. Then 
    \[
        f^0_1 = \alpha_1 \chi_L + (\alpha_2 - \alpha_1) \chi_{L\setminus E_1} + (\alpha_3 - \alpha_2) \chi_{L \setminus (E_1 \cup E_2)} + \dots + (\alpha_n - \alpha_{n - 1}) \chi_{E_n} + (1 - \alpha_n) \chi_{\emptyset}
    \]
    
    Now I claim that
    \[
        f^0_2 = \alpha_1 \chi_{\emptyset} + (\alpha_2 - \alpha_1) \chi_{E_1} + (\alpha_3 - \alpha_2) \chi_{E_1 \cup E_2} + \dots + (\alpha_n - \alpha_{n - 1}) \chi_{L \setminus E_n} + (1 - \alpha_n) \chi_{L}
    \]
    
    Recall that for every $x \in L$, $f^0_2(x) = \beta_i$ for $E_i \ni x$. Using the form of $f^0_2$ above, observe for $x \in E_i$
    \[
        f^0_2(x) = (\alpha_{i+1} - \alpha_i) \chi_{E_1 \cup \dots \cup E_i}(x) + (\alpha_{i+2} - \alpha_{i+1}) \chi_{E_1 \cup \dots \cup E_{i+1}}(x) + \dots + (1 - \alpha_n) \chi_{L}(x) =
    \]
    \[
        = 1 - \alpha_i = \beta_i
    \]

    Note that every pair of indicator functions is defined on a partition, and $\chi_{L \setminus (E_1 \cup \dots \cup E_m)} + \chi_{E_1 \cup \dots \cup E_m} = \chi_L$. Trivially, a telescoping sum $\sum_{i = 1}^{n+1} (\alpha_i - \alpha_{i - 1})$ with $\alpha_0 = 0$ and $\alpha_{n+1} = 1$ is equal to $\alpha_{n+1} = 1$.

    The proof of the generalized version of this theorem for an arbitrary finite number of agents is in the appendix.
\end{proof}

\section{Conclusion}
This paper develops a unified framework establishing the existence of weakly Pareto efficient and envy-free allocations in a broad class of environments. Random allocations are modeled as probability measures on a compact metric space, while agents’ preferences are represented by continuous, concave utility functions defined over the space of probability measures.

The generality of the framework nests several classic existence results in discrete allocation settings with indivisible goods, including the school choice and house allocation problems. The same technique extends naturally to divisible environments such as cake-cutting and land division, where I further show that even under non-atomless preferences, allocation is question can be chosen to have finite support.

Finally, the framework is applied to new classes of allocation problems that lie outside existing models, including the allocation of indivisible goods and services over time and the allocation of differentiated commodities. These applications highlight the flexibility of the approach and its ability to accommodate environments that cannot be captured within standard allocation frameworks.

\printbibliography[heading=bibnumbered]

@article{Foley1967,
  author = {Foley, Duncan K.},
  title = {Resource Allocation and the Public Sector},
  journal = {Yale Economic Essays},
  year = {1967},
  volume = {7},
  pages = {45--90}
}

@article{HyllandZeckhauser1979,
  author = {Hylland, Aanund and Zeckhauser, Richard},
  title = {The efficient allocation of individuals to positions},
  journal = {Journal of Political Economy},
  year = {1979},
  volume = {91},
  pages = {293--313}
}

@article{Berliant1985,
  author  = {Marcus Berliant},
  title   = {An equilibrium existence result for an economy with land},
  journal = {Journal of Mathematical Economics},
  volume  = {14},
  number  = {1},
  pages   = {53--56},
  year    = {1985},
  doi     = {10.1016/0304-4068(85)90026-6}
}

@article{Weller1985,
  author = {Weller, Daniel},
  title = {Fair allocation of indivisible goods and criteria of justice},
  journal = {Econometrica},
  year = {1985},
  volume = {53},
  pages = {1121--1140}
}

@article{Akin1995,
  author = {Akin, Ethan},
  title = {Vilfredo Pareto cuts the cake},
  journal = {Journal of Mathematical Economics},
  year = {1995},
  volume = {24},
  pages = {23--44}
}

@article{Varian1974,
  author = {Varian, Hal R.},
  title = {Equity, Envy, and Efficiency},
  journal = {Journal of Economic Theory},
  year = {1974},
  volume = {9},
  pages = {63--91}
}

@misc{EcheniqueEtAl,
  author = {Echenique, Federico and Miralles, Antonio and Zhang, Jun},
  title = {Fairness and Efficiency for Allocations with Participation Constraints},
  note = {\url{https://authors.library.caltech.edu/102178/}},
  year = {2020}
}

@article{BogomolnaiaMoulin2001,
  author = {Bogomolnaia, Anna and Moulin, Herve},
  title = {A New Solution to the Random Assignment Problem},
  journal = {Journal of Economic Theory},
  year = {2001},
  volume = {100},
  number = {2},
  pages = {295--328}
}

@misc{HalpernShah2021,
  author = {Halpern, Daniel and Shah, Nisarg},
  title = {Fair Division with Subsidy},
  note = {\url{https://arxiv.org/abs/2105.10064}},
  year = {2021}
}

@misc{BudishEtAl2013,
  author = {Budish, Eric and Che, Yeon-Koo and Kojima, Fuhito and Milgrom, Paul},
  title = {Designing Random Allocation Mechanisms: Theory and Applications},
  note = {Working paper},
  year = {2013}
}

@inproceedings{ColeTao2019,
  author = {Cole, Richard and Tao, Xin},
  title = {Mechanism Design for Fair Division: Allocating Divisible Items without Payments},
  journal = {Proceedings of the 2019 ACM Conference on Economics and Computation},
  year = {2019},
  pages = {385--386}
}

@article{Steinhaus1948,
  author = {Steinhaus, Hugo},
  title = {The problem of fair division},
  journal = {Econometrica},
  year = {1948},
  volume = {16},
  pages = {101--104}
}

@article{BerliantThomsonDunz1991,
  author = {Berliant, Marcus and Thomson, William and Dunz, Karl},
  title = {On the fair division of a heterogeneous commodity},
  journal = {Journal of Mathematical Economics},
  year = {1991},
  volume = {20},
  pages = {113--122}
}

@article{HusseinovSagara2013,
  author = {Husseinov, Rauf and Sagara, Norihito},
  title = {Equilibrium and efficiency in fair division with a measure space of agents},
  journal = {Social Choice and Welfare},
  year = {2013},
  volume = {41},
  pages = {45--64}
}

@article{Woodall1980,
  author = {Woodall, Douglas R.},
  title = {Dividing a cake fairly},
  journal = {Mathematical Gazette},
  year = {1980},
  volume = {64},
  pages = {475--477}
}

@book{Alfsen1971,
  author = {Alfsen, Erik M.},
  title = {Compact Convex Sets and Boundary Integrals},
  publisher = {Springer},
  year = {1971},
  note = {\href{https://link.springer.com/book/10.1007/978-3-642-65009-3}{Springer Link}}
}

@book{Border1985,
  author = {Border, Kim C.},
  title = {Fixed Point Theorems with Applications to Economics and Game Theory},
  publisher = {Cambridge University Press},
  year = {1985},
  note = {\url{https://www.amazon.com/Fixed-Point-Theorems-Applns-Border/dp/0521388082}}
}

@article{DvoretzkyWaldWolfowitz1951,
  author = {Dvoretzky, Aryeh and Wald, Abraham and Wolfowitz, Jacob},
  title = {Elimination of Randomization in Certain Statistical Decision Procedures and Zero-Sum Two-Person Games},
  journal = {Annals of Mathematical Statistics},
  year = {1951},
  volume = {22},
  number = {1},
  pages = {1--21}
}

@book{Hildenbrand1974,
  author    = {Werner Hildenbrand},
  title     = {Core and Equilibria of a Large Economy},
  series    = {Princeton Studies in Mathematical Economics},
  number    = {5},
  publisher = {Princeton University Press},
  address   = {Princeton, NJ},
  year      = {1974},
  isbn      = {978-0691645766},
  note      = {Reprinted in the Princeton Legacy Library, 2016}
}

@article{MirallesPycia2021,
  author  = {Antonio Miralles and Marek Pycia},
  title   = {Foundations of Pseudomarkets: Walrasian Equilibria for Discrete Resources},
  journal = {Journal of Economic Theory},
  volume  = {196},
  pages   = {105303},
  year    = {2021},
  doi     = {10.1016/j.jet.2021.105303}
}

@article{MasColell1975,
  author  = {Andreu Mas-Colell},
  title   = {A Model of Equilibrium with Differentiated Commodities},
  journal = {Journal of Mathematical Economics},
  volume  = {2},
  number  = {2},
  pages   = {263--295},
  year    = {1975},
  doi     = {10.1016/0304-4068(75)90028-2},
  note    = {Also available via RePEc / ScienceDirect}
}

@article{PaznerSchmeidler1974,
  author  = {Elisha A. Pazner and David Schmeidler},
  title   = {A Difficulty in the Concept of Fairness},
  journal = {The Review of Economic Studies},
  volume  = {41},
  number  = {3},
  pages   = {441--443},
  year    = {1974},
  doi     = {10.2307/2296762},
  url     = {https://hdl.handle.net/10.2307/2296762}
}

\section{Appendix}

\begin{prop}
    Let $(X,d)$ be a compact metric space. The set of all integer-valued Borel measures on $X$, that are uniformly bounded by an integer $\alpha$, denoted by $\mathcal{A}$, is closed in weak-* topology.
\end{prop}

\begin{prop}
Let $\{x_1,\dots,x_\alpha\}$ be a finite collection of point in $X$. Let $a$ be a measure on $X$. If $\operatorname{supp}(a)\not\subset \{x_1,\dots,x_\alpha\}$, then there exists an open nbhd of a point $N_{\bar x}$ such that
\[
N_{\bar x} \subset \left( B_{\varepsilon_1}(x_1)\cup \cdots \cup B_{\varepsilon_\alpha}(x_\alpha) \right)^C
\]
such that $a(N_{\bar x})>0$.
\end{prop}

\begin{proof}
If $\operatorname{supp}(a)\not\subset \{x_1,\dots,x_\alpha\}$, there exists $\bar x\neq x_k$ for any $k\in\{1,\dots,\alpha\}$ such that every neighborhood of $\bar x$ has a positive measure. Since $X$ is Hausdorff, there exist open $N_{\bar x}$ and $O$ such that $\bar x\in N_{\bar x}$ and $\{x_1,\dots,x_\alpha\}\subset O$ such that $N_{\bar x}\cap O=\emptyset$.
\end{proof}

\begin{prop}
For every open set $O\subset X$, there exists $f\in C(X)$ such that $f\not\equiv 0$ and $\operatorname{supp}(f)\subset O$.
\end{prop}

\begin{proof}
Take an open subset $O\subset X$, $O^C$ is closed. Define
\[
F_\varepsilon = \{x\in X \mid d(x,F)<\varepsilon\}
\qquad\text{and}\qquad
F^\varepsilon = \{x\in X \mid d(x,F)\le \varepsilon\}.
\]
There exists $\varepsilon>0$ such that $(O^{C\varepsilon})^C$ is non-empty. This set is closed and is fully contained in $O$. Using Urysohn's lemma, there exists a continuous $f$ such that
\[
f|_{(O^{C\varepsilon})^C} \equiv 1
\qquad\text{and}\qquad
f|_{O^C=\varepsilon/2} \equiv 0.
\]
Hence, $\operatorname{supp}(f)\subset O$.
\end{proof}

\begin{proof}
    First, observe that for any $a \in \mathcal{A}$, $|supp(a)| \leq \alpha$. Take a converging net $a_i \rightarrow^{w^*} a$, where $a_i \in \mathcal{A}$. Using one of the characterizations of the weak-* topology, it means $\sum_{k = 1}^\alpha a_i(x^i_k)f(x^i_k) \rightarrow af$ for any $f \in C(X)$. Every $a_i$ in the net has at most $\alpha$ atoms, denoted by $\{x^i_1, \dots, x^i_\alpha\}$. Since $X$ is compact, a vector-valued net $(x^i_1, \dots, x^i_\alpha)$ has a converging subnet. Denote its limit by $(x_1, \dots, x_\alpha)$. If $supp(a) \not\subset \{x_1, \dots, x_\alpha\}$, then there exists an open nbhd of a point $N_{\bar{x}}$ such that $N_{\bar{x}} \subset (B_{\epsilon_1}(x_1) \cup \dots \cup B_{\epsilon_\alpha} (x_\alpha))^C$ and $a(N_{\bar{x}}) > 0$. Take $f \in C(X)$ such that $f \not\equiv 0$ and $supp(f) \subset N_{\bar{x}}$. Then $af > 0$ but $\sum_{k = 1}^\alpha a_i(x^i_k)f(x^i_k) \rightarrow 0$.

    WLOG, $\{x_k^{i_l}\}_{l \in N} \rightarrow x_k$ and $\{x_1, \dots, x_\alpha\}$ are distinct points. Then there are $B_{\epsilon_1}(x_1), \dots, B_{\epsilon_\alpha} (x_\alpha)$ pairwise disjoint, such that $x_k^{i_l} \in B_{\epsilon_k}$ residually. Suppose $a(x_k) > 0$ and not an integer, then take $f \not\equiv 0$ such that $f \ \Big |_{B_{\epsilon_k}} \equiv 1$ and $f \ \Big |_{B_{\epsilon_j \neq k}} \equiv 0$. I obtain a contradiction: a net of integers $a_i(x^{i_l}_k)$ converges in topology of real numbers to a strictly positive non-integer.

    If there are duplicates among $\{x_1, \dots, x_\alpha\}$, I still obtain a contradiction: the net becomes $\sum_{k} a_i(x^{i_l}_k)$, but the sum of integers is an integer.
\end{proof}

\begin{lem}
    Let $X$ be a Polish space, and $T: X \rightarrow X$ a continuous operator. Define a push-forward operator $L_T: \Delta(X) \rightarrow \Delta(X)$ as $L_T(p) = p \circ T^{-1}$. Then $L$ is continuous in weak-* topology.
\end{lem}

\begin{proof}
    Since $X$ is a Polish space, $\Delta(X)$ is metrizable and sequential continuity implies continuity. Let $\lambda_n \rightarrow^{w^*} \lambda$, then $(\lambda_n \circ T) f = \lambda_n (f\circ T) \rightarrow \lambda (f\circ T) = (\lambda \circ T^{-1}) f$ for all $f \in C_b(X)$. The equalities follow from the substitution lemma and the continuity of a composition of two continuous functions.
\end{proof}

Vietoris topology on $\mathcal{B}$, the set of measurable subsets of $L$, is defined with a standard subbase. For every non-empty open subset $U \subset L$, define $[U] = \{F \in \mathcal{B}: F \cap U \neq \emptyset\}$, and define $\langle U \rangle = \{F \in \mathcal{B} : F \subset U\}$. The base for this topology is all sets of the form $\langle U, V_1, \dots, V_n \rangle = \langle U \rangle \cap \bigcap_{i = 1}^n [V_n]$, where $U, V_1, \dots V_n$ are non-empty open subsets of $L$.

Recall that $g: L \times \mathcal{B} \rightarrow [0,1]$ is defined as follows
    \[
        g(y,F) = 
        \begin{cases}
        1, & \text{if } y \in F \\
        0, & \text{otherwise}
        \end{cases}
    \]

\begin{lem}
    $g: L \times \mathcal{B} \rightarrow [0,1]$ is measurable with respect to the product Borel sigma algebra on $L \times \mathcal{B}$.
\end{lem}

\begin{proof}
    $g^{-1}(1) = \{(y,F) : y \in L, F \in \mathcal{B}, y \in F\}$ is measurable, since  it is a complement of a measurable set $\{(y,F) : y \in L, F \in \mathcal{B}, y \in F\}^C = \{(y,F): y \in F^C\}$. Observe
    \[
        \{(y,F): y \in F^C\} = \bigcap_{n \in \mathbb{N}} \bigcup_{y \in L} N^y_n (y) \times A(y)
    \]

    where $A(y) = \{F \in \mathcal{B} : y \in F^C\}$, and $\{N^y_n\}_{n 
    \in \N}$ is a countable base of open nbhds of $y$. $A(y)$ is open in Vietoris topology, since $A(y) = \langle L/y \rangle$, where $L/y$ is open since the topology on $L$ is Hausdorff.
\end{proof}

\begin{thm}
    Let $\{f^0_j\}_{j \in J}$ be a family of simple functions such that $f_j \geq 0, \sum_{j \in J} f_j = 1$. Then there exist finitely many families of indicator functions $\{f^*_{k,j}\}_{k \in K,j \in J}$ such that $f^*_{k,j} \geq 0, \sum_{j \in J}f^*_{k,j} = 1$ for every $k$ and such that $(f_1^0, \dots, f_N^0)$ is in the convex hull of $\{(f_{1,1}^*, \dots f_{1,N}^*), \dots, (f_{K,1}^*, \dots f_{K,N}^*)\}$:
    \[ 
        \begin{bmatrix}
        f_1^0 \\
        f_2^0 \\
        \vdots \\
        f_N^0
        \end{bmatrix}
        =
        \sum_{k=1}^{K} a_k
        \begin{bmatrix}
        f_{1,k}^* \\
        f_{2,k}^* \\
        \vdots \\
        f_{N,k}^*
        \end{bmatrix}
    \]
    where $a_k \geq 0, \sum_{k \in K} a_k = 1$.
\end{thm}

\begin{proof}
    For any vector of simple functions $f = (f_1, \dots f_N)$, there is a common partition $E_1, \dots E_I$ of $L$ such that every $f_j$ is constant on every $E_i$. I denote values of the functions on each element of the partition by $f_j^i$.
    
    The proof goes as follows. In the next paragraph, I construct families of indicator functions corresponding to different partitions of $L$. Then, using Farkas' lemma, I claim there exists a vector of coefficients such that the vector $f$ of simple functions is a convex combination of the chosen vectors of indicator functions.

    Recall, $N$ is the number of agents, and $I$ is the number of elements in the partition. For every $m \in \{1, \dots,\min(I,N)\}$, construct a coarser partition with $m$ elements. How many partitions for a fixed $m$? Stirling number of the second kind $S(I,m)$. For every partition, choose $m$ agents from $N$ to give this partition, counting each permutation as a different subset of agents. How many permutations? $P(N,m) = \frac{N!}{(N - m)!}$. The claim now is that there exists $\{a_k^m\}$ such that

    \[
        \begin{bmatrix}
        f_1^0 \\
        f_2^0 \\
        \vdots \\
        f_N^0
        \end{bmatrix}
        =
        \sum_{m=1}^{\min(I,N)} \sum_{k = 1}^{S(I,m) \cdot P(N,m)} a^m_k
        \begin{bmatrix}
        \chi_{A_1^{k(m)}} \\
        \vdots \\
        \chi_{A_j^{k(m)}} \\
        \vdots \\
        \chi_{A_N^{k(m)}}
        \end{bmatrix}
    \]

Let me highlight important characteristics of the sets $(A_1^{k(m)}, \dots, A_N^{k(m)})$. For every $m,k,j$, $A_j^{k(m)} = \bigcup E_i$ for some subset of $(1, \dots, I)$, possibly empty. For every $m$, for every partition in $S(I,m)$ at most $m$ out of $N$ $A_j^{k(m)}$ are non-empty.

I establish the existence of non-negative coefficients $\{a^m_k\}$ that sum to one using Farkas' lemma. Note that in essence, one needs to prove that a system of linear equations $Pa = f$ admits a non-negative solution, where $a = \{a^m_k\}_{m \in M, k \in K}$, $f = \{f^i_j\}_{j \in J, i \in I}$ and $P$ is a matrix of zeros and ones characterized below. I will also show below that $\sum_{m \in M} \sum_{k \in K} a^m_k = 1$ follows from the partition of unity constraint on $f$: $\sum_{j \in J} f^i_j = 1 \ \forall i \in I$.

The rows of $P$ correspond to sets $E_i$ and represent how many times a set $E_i$ gets assigned to a particular agent. The index of a row is $(j,i)$, and there are $J \times I$ rows in total. The entry $(j,i)$ has 1 if in a column partition agent $j$ is allocated a set $E_i$, otherwise it is zero.

The columns represent the assignment of the sets of the original partition $\{E_1, \dots, E_N\}$ in different partitions. The columns are indexed by $(m, Par, Perm)$, where $m$ represents how coarse a partition is; for every $m$, $Par$ gives a particular partition with $m$ elements and $Perm$ gives an assignment of $m$ selected agents to the elements of the partition. $(Par, Perm)$ is indexed by $k(m)$. There are $\sum_{m} S(I,m) \times P(N,m)$ columns.

Let $E_i$ be one of the subsets of the original partition. How many times does the agent $j$ receive this subset in different partitions? Remove the set $E_i$ from the pool of original subsets to be assigned to agents and repeat the procedure for the rest of $E_{k, k \neq i}$. Every row of $P$ features $\sum_{m}^{min(I-1,N)} S(I-1,m) \cdot P(N, m)$ of 1s.


Every column contains $I$ number of 1s, since every column represents how all $(E_1, \dots, E_I)$ are allocated in a corresponding partition, corresponding permutation. Fix column $k(m)$, representing partition and permutation $k(m)$. For every block of $I$ rows that corresponds to agent $j$, for every row $i$, 1 represents that the set $E_i$ is in a subset that is assigned to agent j. Hence, if the set $E_i$ in partition $k(m)$ is allocated to agent $j$, there is 1 in the $i$th row of the agent $j$' block of rows, and there are zeros in the $i$th row of every other agent's block of rows. This description of $P$ contains all the characteristics that are sufficient for the proof of the next claim.

Using Farkas' lemma, if there is no $y$ such that $P^Ty \geq 0$ and $f^Ty < 0$, then I establish the existence of the coefficients in question.

\begin{lem}
    Let $y$ be such that $P^Ty \geq 0$. Then $f^Ty \geq 0$.
\end{lem}

\begin{proof}
    Observe that $f^Ty = \sum_{i \in I} \Big(\sum_{j \in J} f^i_j y_{ji} \Big)$ is the sum of $I$ convex combinations of $\{y_{ji}\}_{j \in J}$.
    If $f^Ty < 0$, then there exists $\tilde{i}$ such that $\sum_{j \in J} f^{\tilde{i}}_j y_{j\tilde{i}} < 0$. WLOG, let's assume $\tilde{i} = 1$. Then there exists a subset of $J$ agents such that $y_{\bar{j}}^1 < 0$ for every such $\bar{j}$. Also WLOG, let's assume agent $1$ is in the subset $J$, so $y_{11} < 0$.

    The idea of the proof is to use the fact that $P^Ty \geq 0$. For every $m$, coarser partition and a particular permutation of agents $(m,k(m))$ $\sum_{j,i} P^T_{(m,k(m)),(j,i)} y_{ji} \geq 0$. Take a subset of these inequalities: select an inequality if for $(m,k(m))$, the column $(1,1)$ has 1 in it. I use this subset of inequalities to bound $|y_{11}|$ from above to ultimately show that $\sum_{i \geq 2} \Big(\sum_{j \in J} f^i_j y_{ji} \Big) \geq |y_{11}|$.

    By varying the column $(j,1)$, we can apply the same construction and upper bound to any $y_{\bar{j}1}$ (and $y_{\bar{j} \tilde{i}}$). To see this, fix $\tilde{i}$. For every allocation $(m,k(m))$ in which agent $k$ receives $E_{\tilde{i}}$, there is an allocation that is the same except $E_{\tilde{i}}$ goes to the agent $l$. Hence, the same $y_{j^\prime i^\prime}$ are picked for the agent $l$, except $y_{l\tilde{i}}$ is replaced with $y_{k\tilde{i}}$. Choosing $\tilde{i} = 1$ and picking agent $j = 1$ is, in fact, without loss.

    Since $P^Ty \geq 0$, list all the rows of $P^T$ where the cell $(1,1)$ is equal to $1$. The number of selected inequalities is equal to the number of times agent 1 gets $E_1$, $\sum_{m}^{min(I-1,N)} S(I,m) \cdot P(N, m)$.
    Every inequality will have $I-1$ of $y_{ji}$ on the left-hand side. Whenever agent $1$ gets $E_1$, no other agent gets it, so inequalities will not feature $y_{j1}, j \neq 1$.
    
    The goal is to show that
    \[
        \sum_{j \in J} f^2_j y^2_j \geq |y_1^1| - \sum_{i \geq 3} \sum_{j \in J} f^i_j y^i_j
    \]

    The procedure is iterative. The main idea is that there are enough inequalities to show first that $\sum_{j \in J} f^2_j y^2_j \geq |y^1_1| - c(I)$, where $c(I)$ is some constant. Iteratively proceed to "extract" elements indexed by $I$ from $c(I)$, getting that $\sum_{j \in J} f^2_j y^2_j \geq |y^1_1| - c(I-1) - \sum_{j \in J} f^I_j y^I_j$ and so on.

    Observe that for every $j \in J$ for the constant $c(I)$ that is independent of $j$, there exists a partition and a permutation to support
    \[
        y^2_j \geq |y^1_1| - c(I)
    \]
    
    There are $J$ inequalities above, and they all hold for any constant $c(I)$ of the type $c(I) = \sum_{i \neq 1,2} y_{j(i)}^i$, $j(i) \in J$ for every $i$. How many constants $c(I)$ of this type are there? We need to allocate $I-2$ sets to $N$ agents, giving $\sum_{m = 1}^{min(I-2,N)} S(I-2,m) \cdot P(N,m)$ constants.

    In particular, for every $j \in J$
    \[
        \sum_{j \in J} f^2_j y^2_j \geq |y^1_1| - c(I-1) - \textcolor{red}{y^{I}_j}
    \]

    where the inequality holds for any constant of type $c(I-1) = \sum_{i \neq 1,2,I} y_{j(i)}^i$, $j(i) \in J$ for every $i$. Similarly to the previous paragraph, there are $\sum_{m = 1}^{min(I-3,N)} S(I-3,m) \cdot P(N,m)$ constants.

    Therefore 
    \[
    \sum_{j \in J} f^2_j y^2_j \geq |y^1_1| - c(I-1) - \sum_{j \in J} f^I_j y_j^I 
    \]

    Iterating, for every $j \in J$

    \[
    \sum_{j \in J} f^2_j y^2_j \geq |y^1_1| - c(I-2) - \textcolor{red}{y_j^{I-1}} - \sum_{j \in J} f^I_j y_j^I 
    \]

    for $c(I-2) = \sum_{i \neq 1,2,I,I-1} y_{j(i)}^i$, $j(i) \in J$ for every $i$.

    Therefore
     \[
    \sum_{j \in J} f^2_j y^2_j \geq |y^1_1| - c(I-2) - \Big(\sum_{j \in J} f^{I-1}_j y_j^{I-1} + \sum_{j \in J} f^I_j y_j^I \Big)
    \]

    Proceeding iteratively, I establish the desired result given that there are enough inequalities to support the claims above. 
    \begin{lem}
        Let $\xi$ be a natural number such that $\xi < I$. For every such $\xi$
        \[
        \sum_{m = 1}^{min(I-\xi,N)} S(I-\xi,m) \cdot P(N,m) = N \times \sum_{m = 1}^{min(I-(\xi + 1),N)} S(I-(\xi + 1),m) \cdot P(N,m)
        \]
    \end{lem}

    \begin{proof}
        Using the recurrence relation of the Stirling numbers of the second kind,
        \[
            S(I-\xi,m) = mS(I - (\xi+1),m) + S(I-(\xi+1),m-1)
        \]
        I get
        \[
        \sum_{m = 1}^{min(I-\xi,N)} S(I-\xi,m) \cdot P(N,m) = \sum_{m = 1}^{min(I-\xi,N)} m S(I-\xi,m) \cdot P(N,m) + 
        \]
        \[
        + \sum_{m = 1}^{min(I-\xi,N)} S(I-(\xi+1),m-1) \cdot P(N,m)
        \]
        Note that $S(n,m) = 0$ where $m > n$, so for all $m > I-(\xi+1)$, $S(I - (\xi+1),m) = 0$. Since $N$ is a natural number, $\sum_{m = 1}^{min(I-\xi,N)} m S(I-(\xi+1),m) \cdot P(N,m) = \sum_{m = 1}^{min(I-(\xi+1),N)} m S(I-(\xi+1),m) \cdot P(N,m)$. 
        For the second term in the sum on the right-hand side, re-labeling produces $\sum_{m = 0}^{min(I-\xi,N)-1} S(I-(\xi+1),m) \cdot P(N,m+1)$. $S(n,0) = 0$ for any $n > 0$, and $P(N,m+1) = (N-m)P(N,m)$. Summing the two terms again produces
        \[
        \sum_{m = 1}^{min(I-(\xi+1),N)} m S(I-(\xi+1),m) \cdot P(N,m) + \sum_{m = 1}^{min(I-\xi,N)-1} S(I-(\xi+1),m) \cdot P(N,m) (N-m) = 
        \]
        \[
        = N \times \sum_{m = 1}^{min(I-(\xi + 1),N)} S(I-(\xi + 1),m) \cdot P(N,m)
        \]
        since if $N \geq 1-\xi$, $min(1-\xi,N) -1 = I -(\xi +1)$, and if $N < 1-\xi$, $min(I - (\xi+1),N) = min(I-\xi,N) = N$.
    \end{proof}

    The claim above establishes the fact that there are just enough inequalities. Every step in the iterative process uses all inequalities from the previous step, and supports the idea used in the construction of the matrix $P$, that is, allocating subsets of the original partition to agents in all possible ways is necessary to apply Farkas' lemma.
\end{proof}

The last claim to prove is $\sum_{m \in M} \sum_{k(m) \in K(m)} a^m_{k(m)} = 1$. To this end, I show that for any $i$, $\sum_{m \in M} \sum_{k(m) \in K(m)} a^m_{k(m)} = \sum_{j \in J} f^i_j$. 
Fix $i$. By construction $\sum_{k,m(k)}P_{(ji),(m,k(m))}a^m_{k(m)} = f^i_j$. Observe that the 1s in a row $(j,i)$ reflect the partitions in which agent $j$ gets $E_i$. For every column, every original subset $E_i$, if agent $k$ gets $E_i$, then no other agent is allocated $E_i$. Thus, every $a^m_{k(m)}$ is present at most once in the sum $\sum_{j \in J} \sum_{k,m(k)}P_{(ji),(m,k(m))}a^m_{k(m)}$. On the other hand, since every $m,k(m)$ assigns an agent to all original subsets, including $E_i$, all $a^m_{k(m)}$ are featured in $\sum_{j \in J} \sum_{k,m(k)}P_{(ji),(m,k(m))}a^m_{k(m)}$.
\end{proof}


\end{document}